\documentclass[a4paper,11pt,reqno]{amsart}
\usepackage[a4paper]{geometry}
\usepackage{amssymb,amsmath}
\usepackage{mathtools}
\usepackage{enumerate}
\usepackage{bookmark}
\usepackage{hyperref}
\usepackage[initials,backrefs]{amsrefs}

\numberwithin{equation}{section}
\theoremstyle{plain}
\newtheorem{theorem}{Theorem}
\newtheorem{lemma}{Lemma}
\theoremstyle{definition}

\newtheorem*{assh*}{(H) Singular distribution hypothesis}
\newcommand{\prob}[1]{\DP\left\{#1\right\}}
\newcommand{\esm}[1]{\mathbb{E}\left[\,#1\,\right]}
\newcommand{\Bone}{\mathbf{1}}

\newcommand{\BC}{\mathbf{C}}

\newcommand{\BG}{\mathbf{G}}
\newcommand{\BH}{\mathbf{H}}

\newcommand{\BP}{\mathbf{P}}
\newcommand{\BU}{\mathbf{U}}
\newcommand{\BV}{\mathbf{V}}

\newcommand{\CJ}{\mathcal{J}}

\newcommand{\DP}{\mathbb{P}}
\newcommand{\DR}{\mathbb{R}}
\newcommand{\DZ}{\mathbb{Z}}
\newcommand{\BDelta}{\mathbf{\Delta}}
\newcommand{\BPsi}{\mathbf{\Psi}}
\newcommand{\Bx}{\mathbf{x}}
\newcommand{\By}{\mathbf{y}}

\newcommand{\FB}{\mathfrak{B}}

\DeclareMathOperator{\dist}{dist}

\newcommand{\ee}{\mathrm{e}}

\newcommand{\condH}{\mathbf{(H)}}

\begin{document}
\title[Wegner bounds for the continuous Bernoulli-Anderson models]{Wegner bounds for one-dimensional multi-particle Bernoulli-Anderson models in the continuum}

\author[T.~Ekanga]{Tr\'esor EKANGA$^{\ast}$}

\address{$^{\ast}$
Institut de Math\'ematiques de Jussieu,
Universit\'e Paris Diderot,
Batiment Sophie Germain,
13 rue Albert Einstein,
75013 Paris,
France}
\email{tresor.ekanga@imj-prg.fr}
\subjclass[2010]{Primary 47B80, 47A75. Secondary 35P10}
\keywords{N-body systems, Wegner bounds, random operators, Anderson localization}
\date{\today}

\begin{abstract}
We prove the Wegner bounds for the one-dimensional interacting multi-particle Anderson models in the continuum. The results apply to singular probability distribution functions such as the Bernoulli's measures. The proofs need the amplitude of the inter-particle interaction potential to be sufficiently weak. As a consequence, the results imply the Anderson localization via the multi-scale analysis.
\end{abstract}

\maketitle

\section{Introduction, assumption and the main results}

 \subsection{Introduction}

In our earlier work \cite{E16}, we proved the Wegner type bounds for the one-dimensional multi-particle interacting Bernoulli-Anderson tight-binding model on the lattice. We aim in the present paper to prove the continuous version of the work.

Note that there is an important number of works on the Wegner estimates see for example the papers \cites{CKM87,CL90,DSS02,CS08,KH13,K08,St01,BCSS10,BK05,W81} for the most famous of them. For different hypotheses, the Wegner bounds were obtain for probability distributions functions with a bounded density in \cites{K08,KH13,CL90} and for H\"older continuous probability distribution functions in \cites{CS08,St01}. We are more interested here with very singular distributions and we recall that the Wegner bounds were established for single-particle models for Bernoulli distributions in one dimension in \cite{CKM87} for the lattice case and in \cite{DSS02} for the continuum case.

In the above papers the main ideas of the proofs used the Furstenberg Theorem for the transfer random matrices and the methods do not apply in higher dimensions. This is why our main strategy in this paper is completely different and investigates new ideas. We show that the single-particle Wegner bounds remain stable when passing to multi-particle models provided that the interaction amplitude is sufficiently weak. The general strategy of our method uses a perturbation argument. 

The main difficulty in the continuous case is that the spectrum of a self-adjoint and compact operator in the continuous space is not necessary finite and in that case, we have to handle an infinite sum. This is done with the help of the Weyl's law. In the first step of the proof, we extend to the non-interacting multi-particle system the known Wegner bound for the single-particle model in one dimension proved by Damanik et al. \cite{DSS02}. The latter bounds are very strong, i.e., the probability decay is exponential and this is the fact that  enables  the  required bound for the multi-particle system. The last step is the obtention of the weakly interacting Wegner bounds from their non-interacting analog using the resolvent identities. 

\subsection{The Anderson model in the continuum}
We fix at the very beginning the number of particles $N\geq 2$. We are concern with multi-particle random Schr\"odinger operators of the following forms:
\[
\BH^{(N)}(\omega):=-\BDelta + \BU+\BV, 
\]
acting in $L^{2}((\DR^{d})^N)$. Sometimes, we will use the identification $(\DR^{d})^N\cong \DR^{Nd}$. Above, $\BDelta$ is the Laplacian on $\DR^{Nd}$, $\BU$ represents the inter-particle interaction which acts as multiplication operator in $L^{2}(\DR^{Nd})$. Additional information on $\BU$ is given in the assumptions. $\BV$ is the multi-particle random external potential also acting as multiplication operator on $L^{2}(\DR^{Nd})$. For $\Bx=(x_1,\ldots,x_N)\in(\DR^{d})^N$, $\BV(\Bx)=V(x_1)+\cdots+ V(x_N)$ and for $x\in\DR^d$,
\[
V(x,\omega)=\sum_{n\in\DZ}q_n(\omega)f(x-n).
\]
The single site potential $f\in L^1(\DR)$ is real-valued supported in $[-1/2;1/2]$ and not $0$ in the $L^1$ sense. The coupling constants $q_n$ are i.i.d. random variables relative to a complete probability space $(\Omega,\FB,\DP)$.                                    

Observe that the non-interacting Hamiltonian $\BH^{(N)}_0(\omega)$ can be written as a tensor product:
\[
\BH^{(N)}_0(\omega):=-\BDelta +\BV=\sum_{k=1}^N \Bone^{\otimes(k-1)}_{L^{2}(\DR^d)}\otimes H^{(1)}(\omega)\otimes \Bone^{\otimes(N-k)}_{L^2(\DR^d)},
\]
where, $H^{(1)}(\omega)=-\Delta + V(x,\omega)$ acting on $L^2(\DR^d)$. We will also consider random Hamiltonian $\BH^{(n)}(\omega)$, $n=1,\ldots,N$ defined similarly. Denote by $|\cdot|$ the max-norm in $\DR^{nd}$. 

 For any $\Bx\in\DZ^{nd}$ and $L>0$, we denote by $\BC^{(n)}_L(\Bx)$, the $n$-particle cube in $L^2(\DR^{nd})$ i.e., $\BC^{(n)}_L(\Bx):=\left\{\By\in\DR^{nd}: |\Bx-\By|\leq L\right\}$. We also denote by $\sigma(\BH^{(n)}_h(\omega))$ the spectrum of $\BH^{(n)}_h(\omega)$ and by $\BG^{(n)}_{\BC^{(n)}_L(\Bx),h}(E)$ the resolvent operator of $\BH^{(n)}_{\BC^{(n)}_L(\Bx),h}(\omega)$ for $E\notin\sigma(\BH^{(n)}_{\BC^{(n)}_L(\Bx),h}(\omega))$.

\subsection{The assumption} 

We denote by $\mu$, the common probability distribution measure of the i.i.d. random variables $\{q_n(\omega); \omega\in\Omega\}$.

\begin{assh*}
We assume that the support of the measure $\mu$ is bounded and not concentrated in a single point and $\int |x|^{\eta}d\mu(x)<\infty$ for some $\eta>0$.
\end{assh*} 

\subsection{The results}
\begin{theorem}[One volume Wegner bounds]\label{thm:FE.Wegner}
Let $d=1$ and $\BC^{(n)}_L(\Bx)$ be an $n$-particle cube in $\DR^{nd}$. Assume that hypothesis $\condH$ holds true. Let $I\subset\DR$ be a compact interval. For any $0<\beta<1$, there exist $L_0=L_0(I,\beta)$ and $h^*=h^*(\|\BU\|,L_0)$ such that for all $h\in(-h^*,h^*)$,
\[
\prob{\dist(E,\sigma(\BH^{(n)}_{\BC^{(n)}_L(\Bx),h}(\omega)))\leq \ee^{-L^{\beta}}}\leq L^{-q},
\]
for all $E\in I$, $L\geq L_0$ and any $q>0$.   
\end{theorem}

\begin{theorem}[Two volumes Wegner bounds]\label{thm:VE.Wegner}
Let $d=1$ and consider two $n$-particle cubes $\BC^{(n)}_L(\Bx)$ and $\BC^{(n)}_L(\By)$ in $\DR^{nd}$. Assume that hypothesis $\condH$ holds true. Then for any $E_0\in \DR$ and any $0<\beta<1$, there exist $L_0=L_0(\beta)>0$, $\delta_0=\delta_0(\|\BU\|,L_0)>0$ such that if we put $I_0:=[E_0-\delta_0,E_0+\delta_0]$, then there exists $h^*=h^*(\|\BU\|,L_0)>0$ such that for all $h\in(-h^*,h^*)$:
\[
\prob{\exists E\in I_0; \max(\dist(E,\sigma(\BH^{(n)}_{\BC^{(n)}_L(\Bx),h}(\omega))),\dist(E,\sigma(\BH^{(n)}_{\BC^{(n)}_L(\By),h}(\omega))))\leq \ee^{-L^{\beta}}}\leq L^{-q},
\]
for all $L\geq L_0$ and any $q>0$.
\end{theorem}

\begin{theorem}[Localization]\label{thm:localization}
Under the hypothesis $\condH$, there exists $h^*>0$ such that for all $h\in(-h^*,h^*)$:
\begin{enumerate}
\item[(i)] the spectrum of the operator $\BH^{(n)}(\omega)$ is pure point with exponentially decaying eigenfunctions at infinity,
\item[(ii)]
the multi-particle Hamiltonian $\BH^{(n)}_h(\omega)$ exhibits complete strong dynamical localization.
\end{enumerate}
\end{theorem}

\section{The road map to the proofs}
We begin with the result on the Wegner estimates for single-particle models established in the paper by Damanik et al. \cite{DSS02}. The estimate is a very strong result, namely the decay bound  is exponential. While for the Anderson localization via the multi-scale analysis some polynomial decay of the probability of the local resolvent is sufficient  to prove  localization. The strategy used in this Section in order to prove our main results is completely different to the proof of \cite{CKM87}. Our idea uses some perturbation argument on the local resolvent identities for operators in Hilbert spaces. While the work \cite{CKM87}, is based on the study of the Lyapounov exponent and the theory of random transfer matrices for single-particle models in one dimension.

\subsection{The single-particle fixed energy Wegner bound}
The known result for single-particle models and for singular probability distributions including the Bernoulli's measures is the following:

\begin{theorem}\label{thm:1p.Wegner}
Let $I\subset \DR$ be a compact interval and $C^{(1)}_L(x)$ be a cube in $\DR^d$. Assume that hypothesis $\condH$ holds true, then for any $0<\beta<1$ and $\sigma>0$, there exist $L_0=L_0(I,\beta,\sigma)>0$ and $\alpha=\alpha(I,\beta,\sigma)>0$ such that,
\[
\prob{\dist(E,\sigma(\BH^{(1)}_{C^{(1)}_L(x)}(\omega)))\leq \ee^{-\sigma L^{\beta}}}\leq \ee^{-\alpha L^{\beta}}
\]
for all $E\in I$ and all $L\geq L_0$.
\end{theorem}
\begin{proof}
We refer the reader to the paper by Damanik et al. \cite{DSS02}.
\end{proof}

\subsection{The non-interacting multi-particle Wegner bound}
We need to introduce first 
\[
 \{ (\lambda^{(i)}_{j_i},\psi^{(i)}_{j_i}): j_i=1,\ldots,|C^{(1)}_L(x_i)|\},
\]
the eigenvalues and the corresponding eigenfunctions of $H^{(1)}_{C^{(1)}_L(x_i)}(\omega)$, $i=1,\ldots,n$. Then, the eigenvalues $E_{j_1\ldots j_n}$ of the non-interacting multi-particle random Hamiltonian $\BH^{(n)}_{\BC^{(n)}_L(\Bx)}(\omega)$ are written as sums:

\[
E_{j_1\ldots j_n}=\sum_{i=1}^{n} \lambda_{j_i}^{(i)}=\lambda_{j_1}^{(1)}+\cdots+\lambda^{((n)}_{j_n},
\]
while the corresponding eigenfunctions $\BPsi_{j_1,\ldots,j_n}$ can be chosen as tensor products
\[
\BPsi_{j_1,\ldots,j_n}=\phi^{(1)}_{j_1}\otimes\cdots\otimes\psi^{(n)}_{j_n}.
\]
We also denote by $\lambda_{\neq i}:=\sum_{\ell=1, \ell\neq i}^n \lambda^{(\ell)}_{j_{\ell}}$. The eigenfunctions of finite volume Hamiltonians are assumed normalized.

The result of this subsection is 
\begin{theorem}\label{thm:np.0.Wegner}
Let $I\subset\DR$ be a compact interval and consider an $n$-particle cube $\BC^{(n)}_L(\Bx)\subset\DR^{nd}$. Assume that hypothesis $\condH$ holds true and that the parameter $h=0$, then for any $0<\beta<1$ and $\sigma>0$, there exist $L_0=L_0(I,\beta,\sigma)>0$ and $\alpha=\alpha(I,\beta,\sigma)>0$ such that
\[
\prob{\dist(E, \sigma(\BH^{(n)}_{\BC^{(n)}_L(\Bx),0}(\omega)))\leq \ee^{-\sigma L^{\beta}}}\leq C(N,d)\cdot \ee^{-\alpha L^{\beta/p}},
\]
for all $E\in I$ and all $L\geq L_0$.
\end{theorem}

Before giving the proof, we have a Lemma.

\begin{lemma}\label{lem:np.0.Wegner}
Let $E\in\DR$ and $\BC^{(n)}_L(\Bx)\subset \DZ^{nd}$. Assume that for any $\sigma>0$, $0<\beta<1$ and $
i=1,\ldots,n$
\[
\dist(E-\lambda_{\neq i}, \sigma(H^{(1)}_{C^{(1)}_L(x_i)}(\omega)))>\ee^{-\sigma L^{\beta}},
\]
then, 
\[
\dist(E,\sigma(\BH^{(n)}_{\BC^{(n)}_L(\Bx)})) > \ee^{-2\sigma L^{\beta}}.
\]
\end{lemma}

\begin{proof}
In absence of the interaction, the multi-particle random Hamiltonian decomposes as 

\[
\BH^{(n)}_{\BC^{(n)}_L(\Bx),0}(\omega)=H^{(1,n)}_{C^{(1)}_L(x_1)}\otimes^{(n-1)}I+\cdots+I^{(n-1)}\otimes H^{(n,n)}_{C^{(1)}_L(x_n)}(\omega).
\]

For the $n$-particle resolvent operator $\BG^{(n)}_{\BC^{(n)}_L(\Bx)}(E)$, we also have the decomposition

\begin{equation}\label{eq:infinite.sum}
\BG^{(n)}_{\BC^{(n)}_L(\Bx)}(E)=\sum_{E_j\in\sigma(\BH^{(n)}_{\BC^{(n)}_L(\Bx)}} \BP_{\BPsi_{\neq j}}\otimes G^{(1)}_{C^{(1)}_L(x_j)}(E-\lambda_{\neq j}),
\end{equation}
where $\BP_{\BPsi_{\neq j}}$ denotes the projection onto the eigenfunction $\BPsi_{\neq j}=\bigotimes_{\ell=1,\ell\neq j}^{n}\varphi_{\ell}$ and $\{\lambda_i : i=1,\ldots,n\}$ are the eigenvalues of $H^{(1)}_{C^{(1)}_L(x_i)}(\omega), i=1,\ldots,n$ respectively and where $\lambda_{\neq j}=\sum_{i\neq j} \lambda_i$. Therefore, assuming all the eigenfunctions of finite volume Hamiltonians normalized, if each cube $C^{(1)}_L(x_i)$, $i=1,\ldots,n$ satisfies:

\[
\dist(E-\lambda_{\neq i}; \sigma(H^{(1)}_{C^{(1)}_L(x_i)}))> \ee^{-\sigma L^{\beta}},
\]

then, 
\[
\BG^{(n)}_{\BC^{(n)}_L(\Bx)}(E)=\left(\sum_{j\geq 1}\right) \BP_{\BPsi_{\neq j}}\otimes G^{(1)}_{C^{(1)}_L(x_j)}(E-\lambda_{\neq j}),\\
\]
Yielding the  following upper bound for the multi-particle resolvent operator:

\[
\|\BG^{(n)}_{\BC^{(n)}_L(\Bx)}(E)\|\leq
\sum_{ j\geq 1} \ee^{-2\widetilde{\mu}_jL}\cdot \ee^{L^{1/2}}
\]
Above, we used the exponential decay (Anderson localization) of the eigenfunctions of single-particle Hamiltonians.
Now, by the Weyl's law there exists $E^*>0$ which can be arbitrarily large such that $\lambda_j\geq E^*$ for all $j\geq j^*=C_{Weyl}|C^{(1)}_L(u_n)|$. Here $|C^{(1)}_L(u_n)|$, denotes the volume of the cube $C^{(1)}_L(u_n)$. Therefore, we can divide the above sum onto two sums as follows:

\[
\sum_{ j\geq 1} \ee^{-2\widetilde{\mu}_jL}\cdot \ee^{L^{1/2}}\leq \left(\sum_{j\leq j^*}+\sum_{j> j^*}\right) \ee^{-2\widetilde{\mu}_jL}\cdot \ee^{L^{1/2}}.
\]
Thus, the infinite sum, can be bounded as follows, provided that the length $L\geq L^*(N,d,C_{Weyl})$ is large enough,
\[
 \sum_{j> j^*}   \ee^{-2\widetilde{\mu}_jL}\cdot \ee^{L^{1/2}}\leq \frac{1}{2}\ee^{2L^{1/2}},
\]
while the finite sum can be bounded by:

\begin{align*}
\sum_{j\leq j^*}\ee^{-2\widetilde{\mu}_jL}\cdot \ee^{L^{1/2}}&\leq C_{Weyl}\cdot|C^{(1)}_L(u)|\cdot\ee^{-2\widetilde{\mu}_jL}\cdot\ee^{L^{1/2}}\\
\leq \frac{1}{2} \ee^{L^{1/2}}\\
\end{align*}
provided that $L> L^{**}(N,d,C_{Weyl})$ is large enough. Finally, we obtain that 
\[
\|\BG^{(n)}_{\BC^{(n)}_L(\Bx)}(E)\|\leq \ee^{2L^{1/2}}.
\]
Which ends the proof.
\end{proof}

\begin{proof}[Proof of Theorem \ref{thm:np.0.Wegner}]
We have two cases:

\begin{enumerate}
\item[Case (a)]
For all $i=1,\ldots,n$, $x_i=x_1$, so that $\Bx=(x_1,\ldots,x_1)$. Using the decomposition of the $n$-particle Hamiltonian without interaction, we can identify

\begin{align*}
\BH^{(n)}_{\BC^{(n)}_L(\Bx),0}(\omega)&= H^{(1)}_{C^{(1)}_L(x_1)}(\omega)+\cdots+H^{(1)}_{C^{(1)}_L(x_1)}(\omega)\\
&=n\cdot H^{(1)}_{C^{(1)}_L(x_1)}(\omega).
\end{align*}
So that $\sigma(\BH^{(n)}_{\BC^{(n)}_L(\Bx)}(\omega))=n\cdot \sigma(H^{(1)}_{C^{(1)}_L(x_1)}(\omega))$. Therefore, 

\begin{align*}
\dist(E,\sigma(\BH^{(n)}_{\BC^{(n)}_L(\Bx)}))&= \dist(E; n\sigma(H^{(1)}_{C^{(1)}_L(x_1)}))\\
&= n\cdot\dist(\frac{E}{n}; \sigma(H^{(1)}_{C^{(1)}_L(x_1)}(\omega))).
\end{align*}

Hence, applying Theorem \ref{thm:1p.Wegner}, we get
\begin{align*}
\prob{\dist(E,\sigma(\BH^{(n)}_{\BC^{(n)}_L(\Bx),0}(\omega)))\leq \ee^{-\sigma L^{\beta}}}&\leq  \prob{n\cdot\dist(\frac{E}{n},\sigma(H^{(1)}_L(x_1)(\omega)))\leq \ee^{-\sigma L^{\beta}}}\\
&\leq \prob{\dist(\frac{E}{n},\sigma(H^{(1)}_{C^{(1)}_L(x_1)}(\omega)))\leq \frac{1}{n} \ee^{-\sigma L^{\beta}}}\\
&\leq \prob{\dist(\frac{E}{n},\sigma(H^{(1)}_{C^{(1)}_L(x_1)}(\omega)))\leq \ee^{-\sigma L^{\beta}}}\\
&\leq \ee^{-\alpha L^{\beta}}.
\end{align*}

\item[ Case (b)]

The single-particle cubes in the product $\BC^{(n)}_L(\Bx)=C^{(1)}_L(x_1)\times\cdots\times C^{(1)}_L(x_n)$ are not all the same. Thus, we begin with $C^{(1)}_L(x_1)$  and we assume that the total number of the single-particle projections that are the same with $C^{(1)}_L(x_1)$ including $C^{(1)}_L(x_1)$ itself is $n^{(1)}_0$ and the one of those different from $C^{(1)}_L(x_1)$ is $n^{(1)}_+$. Clearly, $n^{(1)}_0+n^{(1)}_+=n$. We denote by $C^{(1)}_L(y_1),\ldots,C^{(1)}_L(y_{n^{(1)}_+})$ the latter cubes which are different from $C^{(1)}_L(x_1)$. We denote by $\varPi_{\CJ_1} C^{(1)}_L(x_1)$,
the elements of $C^{(1)}_L(x_1)$ not belonging to $C^{(1)}_L(y_1)$. So that
\[
\varPi_{\CJ_1} C^{(1)}_L(x_1)\cap C^{(1)}_L(y_1)=\emptyset.
\]
We continue the procedure and find a subset $\varPi_{\CJ_2} C^{(1)}_L(x_1)$ such that 
\[
\varPi_{\CJ_2} C^{(1)}_L(x_1)\cap[ C^{(1)}_L(y_2)\setminus \varPi_{\CJ_1} C^{(1)}_L(x_1)]=\emptyset.
\]
At the end, we find $\varPi_{\CJ_{n_+}} C^{(1)}_L(x_1)$ such that

\[
\varPi_{\CJ_{n_+}} C^{(1)}_L(x_1)\cap \left[C^{(1)}_L(y_{n_+})\setminus\left(\bigcup_{i=1}^{n_+-1} \varPi_{\CJ_i} C^{(1)}_L(x_1)\right)\right]=\emptyset.
\]

For all $\ell=1,\ldots, n^{(1)}_+$ we set,  
 \[
    \Lambda^{(1,i)}_L:=\varPi_{\CJ_i} C^{(1)}_L(x_1);\qquad  D^{(1,\ell)}_L:=C^{(1)}_L(y_{\ell})\setminus\left(\bigcup_{i=1}^{\ell-1} \varPi_{\CJ_i} C^{(1)}_L(x_1)\right).
\]
 We therefore obtain
\[
\left[\bigcup_{i=1}^{n^{(1)}_+}\Lambda^{(1,i)}_L \right]\cap \left[ \bigcup_{\ell=1}^{n^{(1)}_+} D^{(1,\ell)}_L\right]=\emptyset
\]
and remark that each term in the above intersection  is non-empty. Now, we denote by $\FB_1$, the sigma algebra:

\[
\FB_1:=\Sigma\left(\left\{ V(x,\omega);  x\in\bigcup_{\ell=1}^{n^{(1)}_+} D^{(1,\ell)}_L\right\}\right).
\]
 Note that, at the $k^{th}$ step, we construct the domain $D^{(k)}$ in a similar way by finding the corresponding domain $\Lambda^{(k)}$  of the $n^{(k)}_+<n^{(k-1)}_+$ remaining cubes in the same way as above. Repeating this procedure, we reduce step by step  the number of the remaining cubes and arrive at the last step by building the domain $D^{(r)}_L$ with the required conditions and with the number $r\leq n$. In other words, the domains $\Lambda^{(k)}_L$ are non-empty and satisfy the conditions: for all $k\neq k'$, $\Lambda^{(k)}_L\cap\Lambda^{(k')}_L=\emptyset$. For all $1\leq k \leq r$, we also define the following sigma algebras:

\[
\FB_{k}:=\Sigma\left(\left\{V(x,\omega): x\in \bigcup_{\ell=1}^{n^{(k)}_+} D^{(k,\ell)}_L\right\}\right)
\]
Next, we define the sigma-algebra:
\[
\FB_{<\infty}:=\bigcup_{k=1}^r \FB_{k}.
\]
Now, observe that for any $i=1,\ldots,n$, we can find a sigma algebra $\FB_{j_i}$ for some $j_i\in\{1,\ldots,r\}$ such that the quantity $\lambda_{\neq i}$ is $\FB_{j_i}$-measurable.  
By Lemma \ref{lem:np.0.Wegner}, we have that

\begin{gather*}
\prob{\dist(E,\sigma(\BH^{(n)}_{\BC^{(n)}_L(\Bx)}(\omega)))\leq \ee^{-\sigma L^{\beta}}}\\
\leq \prob{\exists i=1,\ldots,n:\dist(E-\lambda_{\neq i},\sigma(H^{(1)}_{C^{(1)}_L(x_i)}(\omega)))\leq \ee^{-\frac{\sigma}{2} L^{\beta}}}.\\
\end{gather*}

Next, using Theorem \ref{thm:1p.Wegner} and the sigma-algebra $\FB_{<\infty}$, we finally obtain:

\begin{gather*}
\prob{\exists i=1,\ldots,n: \dist(E-\lambda_{\neq i},\sigma(H^{(1)}_{C^{(1)}_L(x_i)}(\omega)))\leq \ee^{-\frac{\sigma}{2} L^{\beta}}}\\
 \leq \sum_{i=1}^n \esm{\prob{\dist(E-\lambda_{\neq i};\sigma(H^{(1)}_{C^{(1)}_L(x_i)}(\omega)))\leq \ee^{-\frac{\sigma}{2} L^{\beta}}\Big| \FB_{<\infty}}}\\
\leq C(N,d) \cdot \ee^{-\alpha L^{\beta}}.
\end{gather*}
This completes the proof of Theorem \ref{thm:np.0.Wegner}.
\end{enumerate}
\end{proof}

\subsection{ The interacting multi-particle Wegner bound}

The main result is 

\begin{theorem}\label{thm:np.Weak.Wegner}
Let $E\in\DR$ and an $n$-particle cube $\BC^{(n)}_L(\Bx)\subset\DR^{nd}$. Assume that assumption $\condH$ holds true. Then for any $0<\beta<1$ and $\sigma>0$ there exist $L_0=L_0(\beta,\sigma)>0$ such that 
\begin{equation}
\prob{\dist(E,\sigma(\BH^{(n)}_{\BC^{(n)}_L(\Bx),h}(\omega)))\leq \ee^{-\sigma L^{\beta}}}\leq L^{-q},
\end{equation}
for all $L\geq L_0$ and any $q>0$.
\end{theorem}

\begin{proof}
We use the second resolvent identity. Let $E\in\DR$. We have 
\[
\BG^{(n)}_{\BC^{(n)}_L(\Bx),h}(E)=\BG^{(n)}_{\BC^{(n)}_L(\Bx),0}(E) + h\BU \BG^{(n)}_{\BC^{(n)}_L(\Bx),0}(E) \BG^{(n)}_{\BC^{(n)}_L(\Bx),h}(E).
\]
So that
\begin{equation}\label{eq:bound.resolvent}
\|\BG^{(n)}_{\BC^{(n)}_L(\Bx),h}(E)\|\leq \|\BG^{(n)}_{\BC^{(n)}_L(\Bx),0}(E)\|+ |h|\cdot\|\BU\|\|\BG^{(n)}_{\BC^{(n)}_L(\Bx),0}(E)\|\|\BG^{(n)}_{\BC^{(n)}_L(\Bx),h}(E)\|.
\end{equation}
Now, assume that 
\[
\dist(E,\sigma(\BH^{(n)}_{\BC^{(n)}_L(\Bx),0}(\omega))) > \ee^{-\sigma L_0^{\beta}},
\]
then 
\[
\|\BG^{(n)}_{\BC^{(n)}_L(\Bx),0}(E)\| \leq \ee^{\sigma L_0^{\beta}}.
\]
Therefore, it follows from \eqref{eq:bound.resolvent} that:
\[
\|\BG^{(n)}_{\BC^{(n)}_L(\Bx),h}(E)\| \leq \ee^{\sigma L_0^{\beta}}+ |h|\cdot\|\BU\|\cdot\ee^{\sigma L_0^{\beta}}\cdot \|\BG^{(n)}_{\BC^{(n)}_L(\Bx),h}(E)\|.
\]
Next, we choose the parameter $h>0$ in such a way that:
\begin{equation}\label{eq:cond.h}
|h|\|\BU\| \ee^{\sigma L_0^{\beta}}\leq \frac{1}{2}.
\end{equation}
Hence, setting
\[
h^*:=\frac{1}{2\|\BU\|\ee^{\sigma L_0^{\beta}}},
\]
the condition of \eqref{eq:cond.h}  is satisfied for all $h\in(-h^*,h^*)$. Thus, since $L\geq L_0$, we obtain that
\[
\|\BG^{(n)}_{\BC^{(n)}_L(\Bx),h}(E)\|\leq \ee^{\sigma L^{\beta}} + \frac{1}{2}\cdot \|\BG^{(n)}_{\BC^{(n)}_L(\Bx),h}(E)\|.
\]
Yielding,
\[
\|\BG^{(n)}_{\BC^{(n)}_L(\Bx),h}(E)\| \leq \ee^{\sigma L^{\beta}},
\]
for a new $\beta$ bigger than the previous one, provided that $L_0$ is large enough. Finally, using the notations and the result of Theorem \ref{thm:np.0.Wegner}, we bound the probability  for the interacting Hamiltonian as follows,

\begin{align*}
\prob{\dist(E,\sigma(\BH^{(n)}_{\BC^{(n)}_L(\Bx),h}(\omega))) \leq \ee^{-\sigma L^{\beta}}}&\leq \prob{\dist(E,\sigma(\BH^{(n)}_{\BC^{(n)}_L(\Bx),0}(\omega))) \leq \ee^{-\sigma L_0^{\beta}}}\\
&\leq \prob{\dist(E,\sigma(\BH^{(n)}_{\BC^{(n)}_L(\Bx),0}(\omega))) \leq \ee^{-\sigma L^{\beta/p}}}\\
&\leq C(N,d) \cdot \ee^{-\alpha L^{\beta}/p}< L^{-q},
\end{align*}
for all $L\geq L_0$ and any $q>0$, provided that $L_0>0$ chosed as in Theorem \ref{thm:np.0.Wegner} is large enough. Also, we used that for any $L\geq L_0$, we can find an interger $p=p(L,L_0)\geq 1$ such that $L\leq L_0^p$.
\end{proof}

\subsection{The multi-particle variable energy Wegner bound}
In this Section, we prove the variable energy Wegner bound on intervals of small amplitude and the main tool is the resolvent identity. Although they are not easy to obtain, the variable energy Wegner type bounds are useful for the variable energy multi-scale analysis in order to prove the Anderson localization. The result is,

\begin{theorem}\label{thm:np.var.Wegner}
Let $E_0\in\DR$ and a cube $\BC^{(n)}_L(x)$ in $\DR^{nd}$. Assume that hypothesis $\condH$ holds true, then for any $0<\beta<1$ and any $\sigma>0$ there exist $L_0=L_0(E_0,\beta,\sigma)>0$, $\alpha=\alpha(E_0,\beta,\sigma)>0$ and $\delta_{0}=\delta_0(L_0,\beta)>0$  such that:
\[
\prob{\exists E\in[E_0-\delta_0;E_0+\delta_0]: \dist(E,\sigma(\BH^{(n)}_{\BC^{(n)}_L(\Bx)}(\omega)))\leq \ee^{-\sigma L^{\beta}}} \leq L^{-q},
\]
for all $L\geq L_0$ and any $q>0$.
\end{theorem}

\begin{proof}
By the first resolvent equation, we have 

\begin{equation}\label{eq:R.2}
\BG^{(n)}_{\BC^{(n)}_L(\Bx)}(E)= \BG^{(n)}_{\BC^{(n)}_L(\Bx)}(E_0)+ (E-E_0) \BG^{(n)}_{\BC^{(n)}_L(\Bx)}(E) \BG^{(n)}_{\BC^{(n)}_L(\Bx)}(E_0).
\end{equation}
We choose the parameters $\beta$, $L_0$ and $\sigma$ as in Theorem \ref{thm:1p.Wegner}. Now, if $\dist(E_0,\sigma(\BH^{(n)}_{\BC^{(n)}_L(\Bx)}(\omega)))> \ee^{-ç\sigma L_0^{\beta}}$, and $|E-E_0|\leq \frac{1}{2}\ee^{-\sigma L_0^{\beta}}$, then
\[
\dist(E,\sigma(\BH^{(n)}_{\BC^{(n)}_L(\Bx)}(\omega))) > \frac{1}{2} \ee^{-\sigma L^{\beta}}.
\]
Indeed, it follows from \eqref{eq:R.2} that
\begin{align*}
\|\BG^{(n)}_{\BC^{(n)}_L(\Bx)}(E)\|&\leq \|\BG^{(n)}_{\BC^{(n)}_L(\Bx)}(E_0)\|+ |E-E_0|\cdot\|\BG^{(n)}_{\BC^{(n)}_L(\Bx)}(E)\|\cdot\|\BG^{(n)}_{\BC^{(n)}_L(\Bx)}(E_0)\|\\
&\leq \ee^{\sigma L_0^{\beta}} + \frac{1}{2}\cdot \ee^{-\sigma L_0^{\beta}}\cdot\ee^{\sigma L_0^{\beta}}\cdot \|\BG^{(n)}_{\BC^{(n)}_L(\Bx)}(E)\|\\
&\leq \ee^{\sigma L_0^{\beta}} + \frac{1}{2} \|\BG^{(n)}_{\BC^{(n)}_L(\Bx)}(E)\|.
\end{align*}
Yielding, 
\[
\|\BG^{(n)}_{\BC^{(n)}_L(\Bx)}(E)\|-\frac{1}{2}\| \BG^{(n)}_{\BC^{(n)}_L(\Bx)}(E)\|\leq \ee^{\sigma L^{\beta}},
\]
since, $L\geq L_0$. Thus,
\[
\|\BG^{(n)}_{\BC^{(n)}_L(\Bx)}(E)\|\leq 2 \ee^{\sigma L^{\beta}},
\]
or in other terms with a new $\beta$ bigger  than the first one:
\[
\dist(E,\sigma(\BH^{(n)}_{\BC^{(n)}_L(\Bx)}(\omega))) > \ee^{-\sigma L^{\beta}}.
\]
Hence, setting $\delta_0:=\frac{1}{2} \ee^{-\sigma L_0^{\beta}}$, we bound the probability in the statement of Theorem \ref{thm:np.var.Wegner} as follows,

\begin{gather*}
\prob{\text{$|E-E_0|\leq \delta_0$ and $\dist(E,\sigma(\BH^{(n)}_{\BC^{(n)}_L(\Bx)}(\omega))) \leq \ee^{-\sigma L^{\beta}}$}}\\
\leq \prob{\dist(E_0,\sigma(\BH^{(n)}_{\BC^{(n)}_L(\Bx)}(\omega))) \leq \ee^{-\sigma L_0^{\beta}}}.
\end{gather*}
Next, for $L\geq L_0$, we can find  an integer $p=p(L_0,L)\geq 1$ in such a way  that $L_0^p\geq L$. Finally, using Theorem \ref{thm:1p.Wegner}, we obtain
\begin{align*}
&\prob{\dist(E_0,\sigma(\BH^{(n)}_{\BC^{(n)}_L(\Bx)}(\omega))) \leq \ee^{-\sigma L_0^{\beta}}}\\
&\qquad \qquad \leq \prob{\dist(E_0,\sigma(\BH^{(n)}_{\BC^{(n)}_L(\Bx)}(\omega))) \leq \ee^{-\sigma L^{\beta/p}}}\\
&\qquad \qquad \leq L^{-q},
\end{align*}
for all $L\geq L_0$ and any $q>0$.
\end{proof}

\section{Conclusion: proof of the results}

\subsection{Proof of Theorem \ref{thm:FE.Wegner}}
Clearly, the assertion of Theorem \ref{thm:np.var.Wegner} implies Theorem \ref{thm:FE.Wegner} because the probability of the fixed energy one volume Wegner bound is bounded by the one of the variable energy given in the assertion of Theorem \ref{thm:np.var.Wegner} and this completes the proof of Theorem \ref{thm:FE.Wegner}.

\subsection{Proof of Theorem \ref{thm:VE.Wegner}}
The result of Theorem \ref{thm:VE.Wegner} follows from the assertion of Theorem \ref{thm:np.var.Wegner}. Indeed, the probability on the variable energy Wegner bound for two cubes given in Theorem \ref{thm:VE.Wegner} is bounded by that of one of the cubes at the same energy. This proves the result.

\subsection{Proof of Theorem \ref{thm:localization}}
The above one and two volumes Wegner type bounds are sufficient to derive the Anderson localization in both the spectral exponential and the strong dynamical localization under the weak interaction regime of the $N$-body interacting disordered system. More specifically, these estimates play a crucial role in the proofs of the Anderson localization via the multi-scale analysis. In fact, they are used to bounds in probability some resonant effects which appear in the course of the multi-scale analysis when the energy is too much closed to the spectrum of the local random Hamiltonians.

Once the resonances controlled by the Wegner bounds, the rest of the analysis is based on the geometric resolvent inequalities which make a link between the decay (usually exponential) of the matrix elements of the local resolvents of Hamiltonians in bounded domains (actually cubes) with the one of the local resolvents of Hamiltonians in sub-domains (sub-cubes). We refer the reader interested with the use of the many-body Wegner estimates, to the following works on the Anderson localization, \cite{E11}
 for the low energy regime and \cites{E16,E13,E17} for the weak interaction regime.

Finally, the localization results are proved via classical known methods, see for example the papers \cites{E11,E17} for multi-particle systems which ideas go back to the work by von Dreifus and Klein \cite{DK89} for single-particle models.

Localization for singular Bernoulli distributions is then deduced in the present paper. So, the work makes an important contribution in the field of the mathematics of disordered many body quantum systems.

\end{document}